\documentclass[a4paper,onecolumn,accepted=2021-05-28]{quantumarticle}
\pdfoutput=1
\usepackage{tikzit}

\tikzstyle{box}=[shape=rectangle, text height=1.5ex, text depth=0.25ex, yshift=0.5mm, fill=white, draw=black, minimum height=5mm, yshift=-0.5mm, minimum width=5mm, font={\small}]
\tikzstyle{Z dot}=[inner sep=0mm, minimum size=2mm, shape=circle, draw=black, fill={rgb,255: red,221; green,255; blue,221}]
\tikzstyle{Z phase dot}=[minimum size=5mm, font={\footnotesize\boldmath}, shape=rectangle, rounded corners=2mm, inner sep=0.2mm, outer sep=-2mm, scale=0.8, tikzit shape=circle, draw=black, fill={rgb,255: red,221; green,255; blue,221}, tikzit draw=blue]
\tikzstyle{X dot}=[Z dot, shape=circle, draw=black, fill={rgb,255: red,255; green,136; blue,136}]
\tikzstyle{X phase dot}=[Z phase dot, tikzit shape=circle, tikzit draw=blue, fill={rgb,255: red,255; green,136; blue,136}, font={\footnotesize\boldmath}]
\tikzstyle{hadamard}=[fill=yellow, draw=black, shape=rectangle, inner sep=0.6mm, minimum height=1.5mm, minimum width=1.5mm]
\tikzstyle{vertex}=[inner sep=0mm, minimum size=1mm, shape=circle, draw=black, fill=black]
\tikzstyle{vertex set}=[inner sep=0mm, minimum size=1mm, shape=circle, draw=black, fill=white, font={\footnotesize\boldmath}]

\tikzstyle{hadamard edge}=[-, dashed, dash pattern=on 2pt off 0.5pt, thick, draw={rgb,255: red,68; green,136; blue,255}]
\tikzstyle{brace edge}=[-, tikzit draw=blue, decorate, decoration={brace,amplitude=1mm,raise=-1mm}]
\tikzstyle{diredge}=[->]
\tikzstyle{dashed edge}=[-, dashed]
\tikzstyle{arrow}=[->, thick]

\input{zx.tikzdefs}

\usepackage{authblk}
\usepackage{braket}

\usepackage[numbers,sort&compress]{natbib}

\usepackage{amsmath,amsfonts}
\usepackage{amsthm}
\newtheorem{theorem}{Theorem}

\newtheorem{definition}{Definition}
\newtheorem{assumption}{Assumption}
\newtheorem{corollary}{Corollary}

\usepackage{thmtools}
\usepackage{thm-restate}

\usepackage{float}

\begin{document}

\title{Analyzing the barren plateau phenomenon in training quantum neural networks with the ZX-calculus}
\author[1,2]{Chen Zhao}
\email{zhaochen17@mails.ucas.ac.cn}
\author[1,2]{Xiao-Shan Gao}
\email{xgao@mmrc.iss.ac.cn}
\affil[1]{Academy of Mathematics and Systems Science, Chinese Academy of Sciences}
\affil[2]{University of Chinese Academy of Sciences}

\maketitle

\begin{abstract}
In this paper, we propose a general scheme to analyze the gradient vanishing phenomenon, also known as the barren plateau phenomenon, in training quantum neural networks with the ZX-calculus. More precisely, we extend the barren plateaus theorem from unitary 2-design circuits to any parameterized quantum circuits under certain reasonable assumptions. The main technical contribution of this paper is representing certain integrations as ZX-diagrams and computing them with the ZX-calculus. The method is used to analyze four concrete quantum neural networks with different structures. It is shown that, for the hardware efficient ansatz and the MPS-inspired ansatz, there exist barren plateaus, while for the QCNN ansatz and the tree tensor network ansatz, there exists no barren plateau.
\end{abstract}

\section{Introduction} 
\label{sec:introduction}
In recent years, hybrid quantum-classical algorithms are widely used in quantum chemistry~\cite{peruzzo2014variational,kandala2017hardware,cao2019quantum,bauer2020quantum}, combinatorial optimization~\cite{farhi2014quantum,zhou2018quantum}, and quantum machine learning~\cite{liu2018differentiable,lloyd2018quantum,havlivcek2019supervised,schuld2020circuit,benedetti2019parameterized,zhao2019qdnn}. In these hybrid quantum-classical algorithms, the goal is usually training parameterized quantum circuits (PQCs) with classical optimizers. The PQC will be applied to an initial state and then the state will be measured on a quantum device. The classical optimizer will update the parameters of the PQC according to the measurement results. As the PQC can be run on noisy intermediate-scale quantum (NISQ~\cite{preskill2018quantum}) devices, these algorithms are regarded as near-term practical quantum algorithms with potential quantum advantages.

There exist many methods to train PQCs. Some of these are gradient-based~\cite{schuld2019evaluating,mari2020estimating,stokes2020quantum,kubler2020adaptive} and some are not~\cite{nakanishi2020sequential,kubler2020adaptive}. In quantum machine learning, gradient-based methods are widely used. When using gradient-based methods to train PQCs, one may suffer from the barren plateau (BP) phenomenon which was first studied in~\cite{mcclean2018barren}. The BP phenomenon is that the gradient of parameters of the PQC will vanish exponentially in terms of the system size. It was proved that if the PQCs form unitary 2-designs, then the BP phenomenon exists~\cite{mcclean2018barren}.
This result has been extended to the case when the PQCs form approximately 2-designs in~\cite{holmes2021connecting}. The BP phenomenon in PQCs of various structures has been proposed. For PQCs with a brick-like structure, if the PQC has locally 2-design, then the existence of BPs depends on the depth of the circuit and the cost-function~\cite{cerezo2020cost}. Let $n$ be the number of qubits of the PQC. For $\mathrm{poly}(n)$-depth PQCs with a brick-like form, there always exist BPs. For $\log(n)$-depth PQCs with a brick-like form, if the cost-function is global, there exist BPs. Otherwise, there exists no BP when the cost-function is local. Too much entanglement will induce BPs~\cite{marrero2020entanglement,patti2020entanglement}. The BP phenomenon in dissipative quantum neural networks has been studied in~\cite{sharma2020trainability}. And the noise from quantum hardware also causes BPs, which are called noise-induced BPs~\cite{wang2020noise}. Several methods to avoid BPs have been proposed~\cite{Grant2019initialization,Volkoff_2021,holmes2021connecting,patti2020entanglement}.

The above results about the BP phenomenon are obtained under certain assumptions of unitary $2$-design and it is still difficult to analyze the BP phenomenon for PQCs besides those containing $t$-design parts.
In this paper, we develop a general scheme to analyze whether there exist BP phenomena when training a concrete PQC. We focus on BP phenomena induced by the structure of PQCs and noise-induced BPs are not considered in this paper.
The most important tool used in this paper is the \textit{ZX-calculus}, a graphical language for describing and reasoning about quantum processes. The ZX-calculus was developed by Coecke and  Duncan in~\cite{coecke2008interacting,coecke2011interacting}, which has various applications including quantum circuit synthesis~\cite{duncan2020graph,kissinger2019reducing,cowtan2019phase,hanks2020effective}, measurement-based quantum computing~\cite{DuncanMBQC,Backens2020extraction}, quantum error correction~\cite{chancellor2016coherent,horsman2017surgery}, condensed matter physics~\cite{east2020akltstates}, quantum machine learning~\cite{yeung2020diagrammatic}, and quantum natural language processing~\cite{coecke2020foundations}. In the ZX-calculus, the objects under consideration are ZX-diagrams, which consist of two kinds of tensors: Z-spiders and X-spiders. A ZX-diagram can be rewritten with ZX-calculus rules. Moreover, every quantum circuit can be converted into a ZX-diagram.

Let $\vec{\theta} = (\theta_1,\dots,\theta_m)$ be a set of parameters.
To analyze the gradient of a PQC $U(\vec{\theta})$  with respect to a  Hamiltonian $H$, we need to estimate the following expectation and the variance
\begin{equation}
 \label{var_grad}
    \mathbf{E}\left[ \frac{\partial \braket{H}}{\partial \theta_j} \right],\
  \mathbf{Var}\left[\frac{\partial\braket{H}}{\partial \theta_j}\right]
\end{equation}
where $\braket{H}$ is defined in Eq.~\eqref{eq-brH}.
It will be shown that the expectation in Eq.~\eqref{var_grad} is always zero.
The PQC is said to have {\em barren plateaus} if the variance in Eq.~\eqref{var_grad}
vanishes exponentially in terms of the size of the PQC.
The PQC is said to have {\em no barren plateau} or {\em trainable}
if the variance in Eq.~\eqref{var_grad} vanishes polynomially in terms of the size of the PQC.

To estimate the expectation and variance in Eq.~\eqref{var_grad},
we first represent them as ZX-diagrams.
Since the expectation and the variance are integrations, the main technical contribution of this paper is representing these integrations as ZX-diagrams and computing them with the ZX-calculus when the PQC satisfies Assumption~\ref{assumption-PQC}.
More precisely, with the rewriting rules in the ZX-calculus, we prove that Eq.~\eqref{var_grad} is equal to the contraction of a tensor network with a similar structure as the PQC. Hence, the existence of BPs is totally characterized by the scaling property of the tensor network.

We use these techniques to analyze whether there exist BP phenomena in the hardware-efficient ansatz~\cite{kandala2017hardware}, the QCNN ansatz~\cite{cong2019quantum}, the tree tensor network ansatz~\cite{grant2018hierarchical}, and the MPS-inspired ansatz~\cite{PhysRevResearch.1.023025}. We show that there exist BPs in hardware-efficient ansatz and MPS-inspired ansatz, and there exists no BP in the QCNN ansatz and the tree tensor network ansatz.

This paper is organized as follows. A brief introduction to the PQC, the BP phenomenon, and the ZX-calculus will be given in Section~\ref{sec:preliminary}. We will prove the main result that characterizes Eq.~\eqref{var_grad} in Section~\ref{sec:analyzing_the_bp_phenomenon_with_the_zx_calculus}. And the analysis of four concrete PQCs is given in Section~\ref{sec:examples}.

\section{Preliminary} 
\label{sec:preliminary}
\subsection{Hybrid quantum-classical algorithms} 
\label{sub:hybrid_quantum_classical_algorithms}
In a hybrid quantum-classical algorithm, there will be an ansatz, which is a   PQC of the form
\begin{equation}
    U(\vec{\theta}) = \prod_{j=1}^M[U_j(\theta_j)\cdot V_j].
\label{eq-gao1}
\end{equation}
In (\ref{eq-gao1}), $U_j(\theta_j), j=1,\dots,M$ are parameterized gates, such as the rotational gates $R_X,R_Y,R_Z$; and $V_j$ are non-parameterized gates, such as the Hadamard gate $H$ and the CNOT gate. The PQC will be applied to an initial state $\rho_0$ and then the state will be measured. The above procedure, which is  the quantum part of the algorithm, will be run on quantum processors.
Meanwhile, there will be a classical part that consists of classical processors to optimize the parameters of the PQC in the quantum part. A cost-function $L(\vec{\theta})$ will be estimated in the classical part based on the measurement results. Usually, the expectation
\begin{equation}
\label{eq-brH}
\braket{H} = \mathrm{Tr}\left( \rho_0 U^\dagger(\vec{\theta})HU(\vec{\theta}) \right)
\end{equation}
of a given Hamiltonian $H$ will be regarded as the cost-function in many tasks.

\begin{figure}[h]
  \centering
  \input{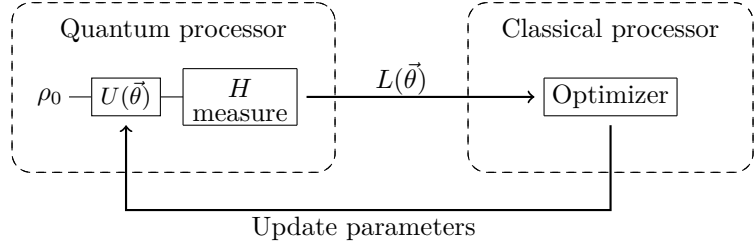}
  \caption{The hybrid quantum-classical algorithm}
  \label{fig:hybrid}
\end{figure}

As demonstrated in Figure~\ref{fig:hybrid}, the quantum part runs the PQC and obtains the measurement results and the classical part estimates the cost-function and updates the parameters. After several iterations, the cost-function may converge and be optimized. Then the training will be stopped. This is the main idea of the hybrid quantum-classical algorithm.

\subsection{Barren plateau phenomenon} 
\label{sub:barren_plateau_phenomenon}
When the parameterized gates are of the form \[
  U_j(\theta_j) = e^{-i \frac{\theta}{2} H_j},
\] where $H_j$ satisfies $H_j^2 = I$, the gradient $\frac{\partial \braket{H}}{\partial \theta_j}$ can be estimated by the parameter shifting rule without changing the structure of the PQC~\cite{schuld2019evaluating}. Once we obtain the gradient, we can use gradient-based optimization methods, such as gradient descent, to optimize the parameters.

Ideally, if the gradient does not vanish too fast as the size of the PQC grows, then the gradient could be estimated efficiently and the PQC could be trained easily. However, the BP phenomenon tells us that in many cases, the gradient vanishes exponentially as the system size grows up. When this happens, the PQC is difficult to be trained. The first rigorous proof of the BP phenomenon is shown below.

\begin{theorem}[\cite{mcclean2018barren}]
  Consider a PQC $U(\vec{\theta}) = V(\theta_M,\dots,\theta_{j+1})U(\theta_j)W(\theta_{j-1},\dots,\theta_{1})$ and a Hamiltonian $H$. The expectation of gradient is 0 if $V$ and $W$ are 1-design. And the variance of gradient $\mathbf{Var}\left[\frac{\partial\braket{H}}{\partial \theta_j}\right]$ vanishes exponentially in terms of the number of qubits if $V$ or $W$ is 2-design.
\end{theorem}

Hence, when designing the ansatz PQC for a hybrid quantum-classical algorithm, we should analyze whether there exist BP phenomena in it to ensure that it is trainable.

\subsection{The ZX-calculus} 
\label{sub:zx_calculus}
We provide a brief introduction to the ZX-calculus. For more details, please refer to~\cite{CKbook,vandewetering2020zxcalculus}.

In the ZX-calculus, quantum states and their transformations are represented as ZX-diagrams which consist of two kinds of tensors: \textit{Z-spiders} and \textit{X-spiders}. A Z-spider is denoted as a green node, and an X-spider is denoted as a red node. They can be written explicitly in the Dirac notation as follows.
\begin{equation}
    \input{spiders.tikz}
    \label{eq:def-spider}
\end{equation}
For a spider, the edges on the left-hand side are called \textit{input} and the edges on the right-hand side are called \textit{output}. The angle $\theta$ is called the \textit{phase} of the spider. For simplicity, we will omit the phase when it is zero. Spiders can be connected with wires. Hence, ZX-diagrams can be regarded as tensor networks generated with Z-spiders and X-spiders. For example, we can use ZX-diagrams to represent the following quantum states and quantum gates.
\begin{equation}
  \input{gates.tikz}
  \label{eq:gates_to_zx}
\end{equation}
Here we introduce a new notation, the yellow box, to represent the Hadamard gate \[
  H = \frac{1}{\sqrt{2}}\begin{pmatrix}
    1 & 1 \\ 1 & -1
  \end{pmatrix}.
\]
Since the gates set $\{R_Z, R_X, H, \mathrm{CNOT}\}$ is universal for quantum computing, in principle, one can convert every quantum circuit to a ZX-diagram with the equations in Eq.~\eqref{eq:gates_to_zx}.

Moreover, the ZX-calculus is a powerful tool for reasoning. There are several rewriting rules in the ZX-calculus with which one can rewrite a ZX-diagram to another equivalent form. Figure~\ref{fig:zx_rules} gives some basic rewriting rules\footnote{\FusionRule is for the ``fusion'' rule; \HadamardRule is for the ``Hadamard color changing'' rule;
\IdRule and \IdIIRule are ``identity'' rules; \PiRule is for the ``$\pi$-copy'' rule; \CopyRule is for the ``copy'' rule; \BialgRule is for the ``bi-algebra'' rule.} in the ZX-calculus. Here, two ZX-diagrams $A,B$ are said to be equivalent if and only if there exists a non-zero constant $c\in \mathbb{C}$, such that {$[A] = c\cdot [B]$, where $[A]$ is the matrix corresponding to the ZX-diagram $A$.}
\begin{figure}[H]
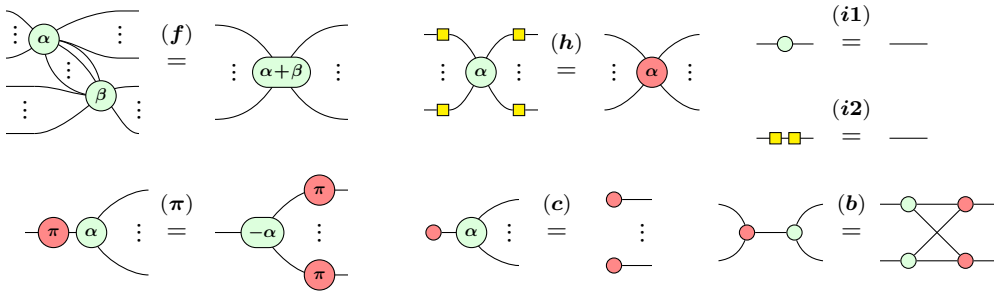

\ctikzfig{ZX-rules}
\caption{Some basic rewriting rules in the ZX-calculus. Here, ``$\dots$'' means 0 or more (figure from~\cite{duncan2020graph}). }
\label{fig:zx_rules}
\end{figure}

Note that the ZX-calculus is \textit{universal}. It means that any linear transformations can be represented as ZX-diagrams. Moreover, the rules in Figure~\ref{fig:zx_rules} are \textit{complete} for the stabilizer quantum mechanics where phases can only be multiples of $\frac{\pi}{2}$~\cite{Backens1,Backens:2015aa}. That is, if two ZX-diagrams are equivalent, then there exists a set of rewriting rules in Figure~\ref{fig:zx_rules} that rewrites one into another. There are also completeness results for the Clifford+T quantum mechanics, where phases can be multiples of $\frac{\pi}{4}$, and for arbitrary ZX-diagrams~\cite{SimonCompleteness,JPV-universal,10.1145/3209108.3209128,wang_completeness_2018,jeandel2019completeness}.

In this paper, we will focus on a canonical form of the ZX-diagram, the \textit{graph-like ZX-diagram} which is defined in~\cite{duncan2020graph}.
\begin{definition}[\cite{duncan2020graph}]
A ZX-diagram is called {\em graph-like} if
  \begin{enumerate}
    \item All spiders are Z-spiders.
    \item Z-spiders are only connected via Hadamard edges.
    \item There exist no parallel Hadamard edges or self-loops.
    \item Every input or output is connected to a Z-spider and every Z-spider is connected to at most one input or output.
  \end{enumerate}
\end{definition}
Two spiders being connected via a Hadamard edge means that they are connected with a Hadamard box. Alternatively, we will also use the dashed blue edge to represent a Hadamard edge.
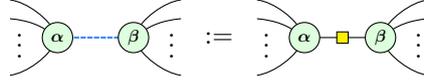
\ctikzfig{h-edge}
All X-spiders can be rewritten to Z-spiders by using the rule \HadamardRule in Figure~\ref{fig:zx_rules}. Connected Hadamard boxes can be canceled with the rule \HHRule and normal edges can be canceled with the rule \FusionRule. Furthermore, parallel Hadamard edges and self-loops can be canceled with rules\footnote{\ParaHadRule is for the ``Hopf'' rule; \SelfLoopRule is for the ``self-loop canceling'' rule; \HadSelfLoopRule is for the ``Hadamard self-loop canceling'' rule.} in Figure~\ref{fig:edge-cancel}. Hence, every ZX-diagram is equivalent to a graph-like ZX-diagram~\cite{duncan2020graph}.
\begin{figure}[h]
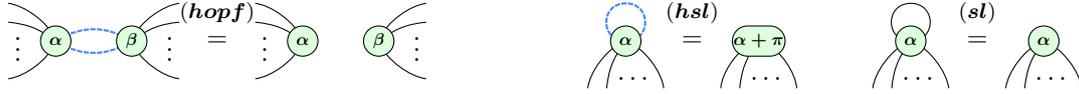

  \ctikzfig{edge-cancel}
  \caption{Rules for canceling parallel edges and self-loops~\cite{duncan2020graph}. }
  \label{fig:edge-cancel}
\end{figure}

\section{Analyzing the BP phenomenon with the ZX-calculus} 
\label{sec:analyzing_the_bp_phenomenon_with_the_zx_calculus}
In this section, we will show how to analyze the BP phenomenon with the ZX-calculus. More precisely, we will show how to estimate the expectation and the variance of the gradient of the cost function of a PQC with respect to a Hamiltonian with the ZX-calculus.
The main technique we used is to compute integration over unitarians with the ZX-calculus.

Scalars are ignored in the rules in Section~\ref{sub:zx_calculus}. However, to consider the BP phenomenon, the scalar is necessary. {Hence, by using the definition of the Z-spider and the X-spider in Eq.~\ref{eq:def-spider},} we can obtain the precise rules with scalars in Figure~\ref{fig:zx_rules_scalar}.
\begin{figure}[h]
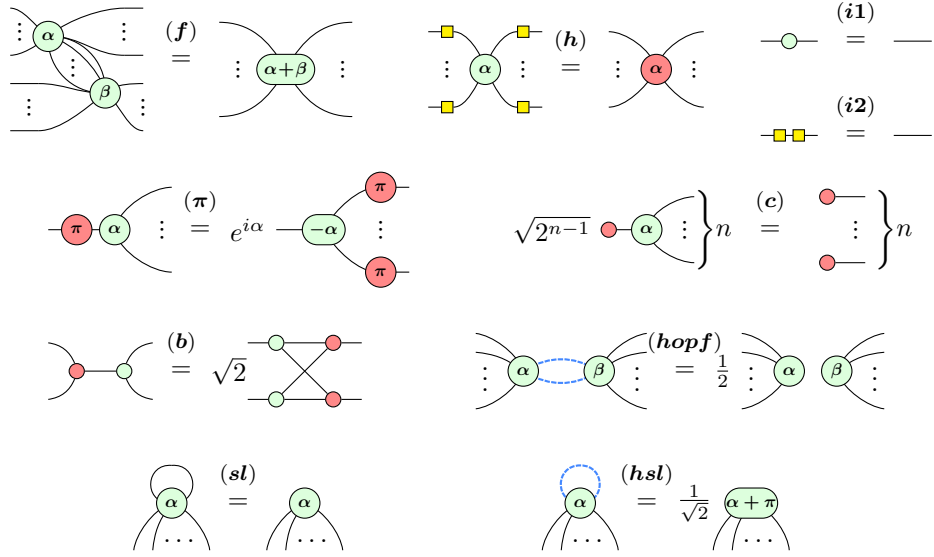

  \ctikzfig{ZX-rules-scalar}
  \caption{Rewriting rules with scalars.}
  \label{fig:zx_rules_scalar}
\end{figure}

In this paper, we consider PQCs under the following assumptions.
\begin{assumption}
    The PQC $U(\vec{\theta})$ satisfies
    \begin{enumerate}
        \item Each gate in $U$ is one of $\{R_X,R_Z,H,\mathrm{CNOT}\} $.
        \item The parameters in $\vec{\theta} = (\theta_1,\dots,\theta_m)$ are independent uniform random variables in the interval $[-\pi,\pi]$.
    \end{enumerate}
    \label{assumption-PQC}
\end{assumption}
{We remark that in the case that a quantum circuit contains gates not satisfying Assumption~\ref{assumption-PQC}, if one can represent these gates as composition of gates which satisfy Assumption~\ref{assumption-PQC}, then the results of this paper still hold. For example, we will first represent $R_Y$  using $R_Z$ and $R_X$ in Section~\ref{sub:tree_tensor_network_ansatz}.}

\subsection{Representing gradients as ZX-diagrams}%
\label{sub:representing_gradients_as_zx_diagrams}

Consider a PQC $U(\vec{\theta})$ of $n$-qubits and a Hamiltonian $H$. {We assume that the input state is pure.} Without loss of generality, we can also assume that we apply this PQC to an initial state $\ket{0}$. Then the expectation of $H$ can be expressed as \begin{equation}
  \braket{H} = \bra{0}U^{\dagger}(\vec{\theta})HU(\vec{\theta})\ket{0}.
\end{equation}

As shown in Section~\ref{sub:zx_calculus}, we can convert the PQC $U(\vec{\theta})$ to a parameterized graph-like ZX-diagram $G_U(\vec{\theta})$ with Eq.~\eqref{eq:gates_to_zx}. Suppose that $$U(\vec{\theta})=c\cdot [G_U(\vec{\theta})]$$
for a constant $c$, then $\braket{H}$ can also be expressed as a ZX-diagram as demonstrated in the following equation. Here, $U(\vec{\theta})$ is under the Assumption~\ref{assumption-PQC}.
\begin{equation}
  \input{expectation.tikz}
  \label{eq:zx-expectation}
\end{equation}

{We remark that in the general case, the input state can be a mixed state $\rho$. Because the ZX-calculus is universal, we can represent $\rho$ as a ZX-diagram $D_{\rho}$. Then by replacing the X-spiders representing zero states on the left- and right-hand sides in Eq.~\eqref{eq:zx-expectation} with $D_{\rho}$, we can still obtain a ZX-diagram representing the expectation $\braket{H}=\mathrm{Tr}(\rho U^{\dagger}HU)$. And results in this paper still hold in this case. }

If we expand the spider by the definition of the Z-spider, we can prove that the gradient $\frac{\partial \braket{H}}{\partial \theta_j}$ can be represented as a ZX-diagram.
\begin{restatable}{theorem}{zxgrad}
  The gradient can be represented as the following equation.
  \ctikzfig{gradient}
  \label{thm:zx-gradient}
\end{restatable}
\begin{proof}
{
We expand the corresponding Z-spiders according to the definition. There will be four terms on the left-hand side. Two of these terms are constants, and thus they will become $0$ after taking the derivative. According to the definition of the X-spider, we can obtain the ZX-diagram on the right-hand side.
The complete proof is given in Appendix~\ref{sec:proof_of_theorem_2}.
}
\end{proof}

{This theorem also gives a graphical proof of the parameter-shift rule in~\cite{schuld2019evaluating}. We will demonstrate it with the following example. Consider the following ansatz. We first can convert it to an equivalent ZX-diagram.}
\[
    \scalebox{1}{\input{para-shift-circuit.tikz}}
\]
{And the expectation $\braket{H}$ of a Hamiltonian $H$ can be represented as the following ZX-diagram.}
\[
    \scalebox{1}{\input{para-shift-expectation.tikz}}
\]
{Then by Theorem~\ref{thmt@@zxgrad}, we can obtain the gradient directly.}
\[
    \scalebox{1}{\input{para-shift-gradient.tikz}}
\]
{Then we can use the definition of the X-spider to obtain the parameter-shift rule as shown below.}
\[
    \scalebox{1}{\input{para-shift-gradient-2.tikz}}
\]
{Here $\braket{H}_{\theta_1,\pm}$ means replacing the parameter $\theta_1$ with $\theta_1\pm \frac{\pi}{2}$ in the ansatz.
}

To analyze the BP phenomenon, we need to compute the expectation \begin{equation}
    \mathbf{E}(\frac{\partial \braket{H}}{\partial \theta_j}) = \int_{\vec{\theta}}p(\vec{\theta})\frac{\partial \braket{H}}{\partial \theta_j} \mathrm{d}\vec{\theta},
    \label{eq:gradient-expectation}
\end{equation} and the variance \begin{equation}
  \mathbf{Var}(\frac{\partial \braket{H}}{\partial \theta_j}) = \mathbf{E}(\left|\frac{\partial \braket{H}}{\partial \theta_j}\right|^2) - \left(\mathbf{E}(\frac{\partial \braket{H}}{\partial \theta_j})\right)^2 = \int_{\vec{\theta}}p(\vec{\theta})\left|\frac{\partial \braket{H}}{\partial \theta_j}\right|^2 \mathrm{d}\vec{\theta} - \left(\mathbf{E}(\frac{\partial \braket{H}}{\partial \theta_j})\right)^2,
  \label{eq:gradient-variance}
\end{equation} for $j=1,\dots,m$.
Here $p(\vec{\theta})$ is the probability of the parameters $\vec{\theta}$.

By Assumption~\ref{assumption-PQC}, Eq.~\eqref{eq:gradient-expectation} and Eq.~\eqref{eq:gradient-variance} can be written as
\begin{equation}
  \mathbf{E}(\frac{\partial \braket{H}}{\partial \theta_j}) = \frac{1}{(2\pi)^m}\int_{\theta_1}\dots\int_{\theta_m}\frac{\partial \braket{H}}{\partial \theta_j} \mathrm{d}\theta_1\dots\mathrm{d}\theta_m,
  \label{eq:gradient-expectation-id}
\end{equation} and
\begin{equation}
  \mathbf{Var}(\frac{\partial \braket{H}}{\partial \theta_j}) = \frac{1}{(2\pi)^m}\int_{\theta_1}\dots\int_{\theta_m}\left|\frac{\partial \braket{H}}{\partial \theta_j}\right|^2 \mathrm{d}\theta_1\dots\mathrm{d}\theta_m - \left(\mathbf{E}(\frac{\partial \braket{H}}{\partial \theta_j})\right)^2,
  \label{eq:gradient-variance-id}
\end{equation} for $j=1,\dots,m$.

We will compute the expectation and variance of the gradients in the next two sections.

\subsection{The expectation of gradients} 
\label{sub:the_expectation_value_of_gradients}

In this section, we will compute the expectation in Eq.~\eqref{eq:gradient-expectation-id}.
As shown in Theorem~\ref{thm:zx-gradient}, the integration \[
  \frac{1}{2\pi}\int_{\theta_k}\frac{\partial \braket{H}}{\partial \theta_j}\mathrm{d}\theta_k,
\] for $k=1,\dots,m$, is also an integration of a ZX-diagram over its parameter $\theta_k$. With the following lemma, the integration can be represented as a ZX-diagram again.

\begin{restatable}{lemma}{intone}
  The following equation holds.
  \ctikzfig{integration-1}
  \label{thm:integration-1}
\end{restatable}
\begin{proof}
{Expand the Z-spiders according to its definition. There will be four terms on the left-hand side. By using \[
        \frac{1}{2 \pi}\int_{-\pi}^{\pi}e^{i \alpha} \mathrm{d}\alpha = 0,
    \]
only two terms left. We will use the definition of the Z-spider again to obtain the ZX-diagram on the right-hand side.}
The complete proof is given in Appendix~\ref{sec:proof_of_lemma_1}.
\end{proof}

With this lemma, {we give a graphical proof of the following theorem, which is also proved in~\cite{holmes2021connecting}.}
\begin{theorem}
    Under Assumption~\ref{assumption-PQC}, the integration \[
  \frac{1}{2\pi}\int_{\theta_j}\frac{\partial \braket{H}}{\partial \theta_j}\mathrm{d}\theta_j = 0.
\]
\end{theorem}
\begin{proof}
Using the relation in Lemma~\ref{thm:integration-1} on the ZX-diagram in Theorem~\ref{thm:zx-gradient}, we have
\ctikzfig{integration-theta-j}
\end{proof}

As a corollary, the expectation of the gradient in Eq.~\eqref{eq:gradient-expectation-id} is zero.
\begin{corollary}
    Under Assumption~\ref{assumption-PQC}, the expectation \[
        \mathbf{E}(\frac{\partial \braket{H}}{\partial \theta_j}) = \frac{1}{(2\pi)^m}\int_{\theta_1}\dots\int_{\theta_m}\frac{\partial \braket{H}}{\partial \theta_j} \mathrm{d}\theta_1\dots\mathrm{d}\theta_m = 0,\text{ for } j=1,\dots,m.
  \]
\end{corollary}

\subsection{The variance of gradients} 
\label{sub:the_variance_of_gradients}

In this section, we will compute the variance in Eq.~\eqref{eq:gradient-variance}.

Because the gradient $\frac{\partial \braket{H}}{\partial \theta_j}$ is a real number and the expectation is 0, by Eq.~\eqref{eq:gradient-variance}, the variance is the expectation of $\left(\frac{\partial \braket{H}}{\partial \theta_j}\right)^2$, which can be represented as follows by Theorem~\ref{thmt@@zxgrad}.
\begin{equation}
  \input{gradient-2.tikz}
  \label{eq:gradient-2}
\end{equation}

Similar to Lemma~\ref{thm:integration-1}, we can prove the following lemma.
\begin{restatable}{lemma}{inttwo}
  The following equation holds.
  \begin{figure}[h]
    \centering
    \scalebox{0.6}{\input{integration-2.tikz}}
  \end{figure}
  \label{thm:integration-2}
\end{restatable}
\begin{proof}
{We expand all Z-spiders and there will be 16 terms on the left-hand side. Again, using the relation \[
        \frac{1}{2 \pi}\int_{-\pi}^{\pi}e^{i \alpha} \mathrm{d} \alpha = 0,
    \]
6 terms was left. We can use the definition of the Z-spider again to obtain the right-hand side.
The complete proof is given in Appendix~\ref{sec:proof_of_lemma_2}.}
\end{proof}

There exist three terms after integration. Hence, computing the variance of gradients is much more complicated than computing the expectation.
We denote the three ZX-diagrams in Lemma~\ref{thm:integration-2} as
\begin{equation}
  \input{T1-T2-T3.tikz}
  \label{eq:T1-T2-T3}
\end{equation}
And we introduce a new notation
\[
    V_{U}^{a_1,\dots,a_m},\ a_j\in\{T_1,T_2,T_3\},
\]
to represent the following ZX-diagram.
\begin{equation}
    \scalebox{0.9}{\input{V-def.tikz}}
  \label{eq:V-def}
\end{equation}
Here, $U(\theta_1,\dots,\theta_m)$ is a PQC with $m$ parameters and $G_U$ is the graph-like ZX-diagram corresponding to $U$. With this notation, we have the following theorem.
\begin{theorem}
    Under Assumption~\ref{assumption-PQC}, the following equation holds.
    \label{thm:var_theta_j}
\end{theorem}
\[
  \mathbf{Var}\left(\frac{\partial\braket{H}}{\partial \theta_j}\right) = \frac{|c|^2}{4^n}\cdot\sum_{a_k\in\{T_1,T_2,T_3\},\ k\neq j} V_U^{a_1,\dots,a_{j-1},T_2,a_{j+1},\dots,a_m}.
\]
\begin{proof}
By Eq.~\eqref{eq:gradient-2} and Lemma~\ref{thm:integration-2}, we have the following equation.
\ctikzfig{integration-2-theta-j-1}
And by the following equations,
\[
  \scalebox{0.8}{\input{integration-2-theta-j-proof.tikz}}
\]
we have
\ctikzfig{integration-2-theta-j-2}
Then, by the definition of $V_{U}^{a_1,\dots,a_m}$, we obtain
\[
  \mathbf{Var}\left(\frac{\partial\braket{H}}{\partial \theta_j}\right) = \frac{|c|^2}{4^n}\cdot\sum_{a_k\in\{T_1,T_2,T_3\},\ k\neq j} V_U^{a_1,\dots,a_{j-1},T_2,a_{j+1},\dots,a_m}.
\]
\end{proof}

Hence, to compute the variance, we need to sum over $3^{m-1}$ terms of the tensor \[
  V_U^{a_1,\dots,a_{j-1},T_2,a_{j+1},\dots,a_m}.
\] It seems inaccessible when $m$ is large. But in many cases, we have simple ways to compute this sum.

{Recall that we have converted quantum circuits to graph-like ZX-diagrams. Hence, spiders are connected with Hadamard edges.} Let us consider two spiders $W_j,W_k$ corresponding to the parameters $\theta_j,\theta_k$ in $G_U$. Suppose that $W_j$ and $W_k$ are connected with a Hadamard edge. Then by the following lemma, the Hadamard edge can be removed after integration over $\theta_j$ and $\theta_k$.
\begin{restatable}{lemma}{inthad}
The following equation holds.
\[
  \scalebox{0.75}{\input{integration-hadamard.tikz}}
\]
\label{thm:integration-hadamard}
Here $M_{a_j,a_k}$ is the element on the $a_j$-th row and $a_k$-th column in the following $3\times 3$ matrix \[
  M = \frac{1}{4}\begin{pmatrix}
  1 & 1 & 1 \\ 1 & 1 & -1 \\ 1 & -1 & 1
  \end{pmatrix}
\]
\end{restatable}
\begin{proof}
    {By Lemma~\ref{thmt@@inttwo}, there will be 9 terms on the left-hand side. We can use rewriting rules in Figure~\ref{fig:zx_rules_scalar} on each term to remove Hadamard edges to obtain the form on the right-hand side.
    The complete proof is given in Appendix~\ref{sec:proof_of_lemma_3}.}
\end{proof}

Applying this lemma to the variance recursively, we can remove all the Hadamard edges connecting two parameterized spiders. And the big tensor \[
  V_U^{a_1,\dots,a_{j-1},T_2,a_{j+1},\dots,a_m}
\] will be broken into smaller tensors that are connected with $M$. It is a new tensor network whose structure is similar to $G_U$. To compute the variance, the only thing we need to do is contracting this new tensor network. Figure~\ref{fig:variance-tensor-network} demonstrates the above procedure for the case that all spiders in $G_U$ are parameterized and are connected with Hadamard edges.
\begin{figure}[h]
  \centering
  \scalebox{1}{\input{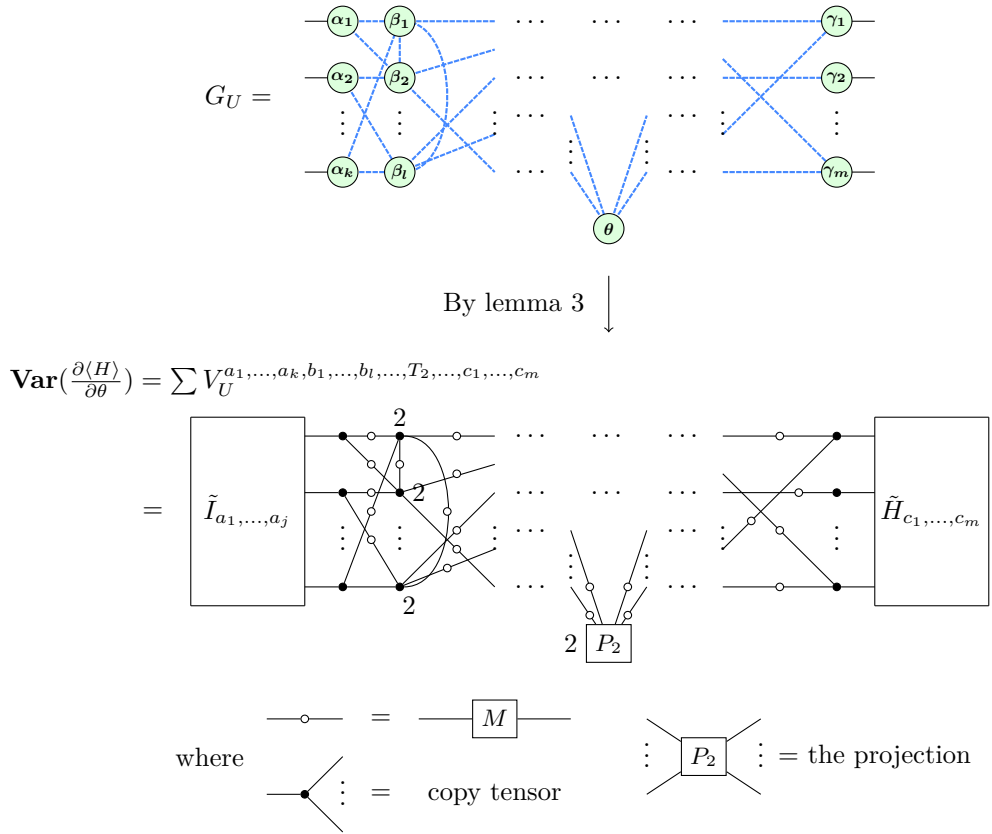}}
  \caption{Computing the variance with the tensor network}
  \label{fig:variance-tensor-network}
\end{figure}
The tensor $\tilde{I}_{a_1,\dots,a_k}$ is related to the input state, while the tensor $\tilde{H}_{c_1,\dots,c_m}$ is related to the Hamiltonian $H$. And $P_2$ is a projection that has only one non-zero entry. That is \[
  P_2(2,2,\dots,2) = 1.
\]
Also, note that there is a scalar $2$ for each internal copy tensor. This scalar comes from the following equation.
\begin{equation}
  \input{T1-T2-T3-scalar.tikz}
\end{equation}

{In Appendix~\ref{appendix:a_simple_example}, a simple example is given to illustrate the techniques introduced in this paper.}

In conclusion, computing the variance of gradients is reduced to contracting a tensor network corresponding to the circuit. In the next section, we will use these techniques to analyze the BP phenomenon for several for concrete PQCs.
%

\section{Analyzing the BP phenomenon for four PQCs} 
\label{sec:examples}

In this section, we will analyze the BP phenomenon for four PQCs with the techniques introduced in Section~\ref{sec:analyzing_the_bp_phenomenon_with_the_zx_calculus}.

\subsection{Hardware-efficient ansatz} 
\label{sub:hardware_efficient_ansatz}

Consider a hardware-efficient ansatz~\cite{kandala2017hardware} PQC of the following form.
\begin{equation}
  \scalebox{0.75}{\input{hardware-efficient-circuit.tikz}}
  \label{eq:hardware-efficient-circuit}
\end{equation}

{The first step is converting it to a graph-like ZX-diagram. }
Suppose the circuit is of $n$-qubits. By the conversion rules in Eq.~\eqref{eq:gates_to_zx} and rewriting rules in Figure~\ref{fig:zx_rules_scalar}, we obtain a graph-like ZX-diagram where most spiders are parameterized.
\[
  \scalebox{1}{\input{hardware-efficient-zx.tikz}}
\]
Here, the Z-spider with ``$\dots$'' represents a Z-spider with a parameter.

{Then we are going to represent the variance of gradients as a tensor network.}
By Lemma~\ref{thm:integration-hadamard}, we can remove most Hadamard edges. Then the remaining part consists of \[
  \scalebox{1}{\input{hardware-efficient-entangler.tikz}}.
\]
By similar techniques used in Lemma~\ref{thm:integration-hadamard}, we can prove the following lemma.
\begin{restatable}{lemma}{intcnot}
The following equation holds.
\[
  \scalebox{0.75}{\input{integration-entangler.tikz}}
\]
\label{thm:integration-entangler}
Here $ET_{a_1,a_2,a_3}$ is an element in the following $3\times 3\times 3$ tensor \[
  ET[1,\cdot,\cdot] = \frac{1}{8}\begin{pmatrix}
    1 & 0 & 0 \\ 0 & 1 & 0 \\ 0 & 0 & 1
  \end{pmatrix},\quad ET[2,\cdot,\cdot] = \frac{1}{8}\begin{pmatrix}
    0 & 1 & 0 \\ 1 & 0 & 0 \\ 0 & 0 & 0
  \end{pmatrix},\quad ET[3,\cdot,\cdot] = \frac{1}{8}\begin{pmatrix}
    0 & 0 & 1 \\ 0 & 0 & 0 \\ 1 & 0 & 0
  \end{pmatrix}.
\]
\end{restatable}
\begin{proof}
    Refer to Appendix~\ref{sec:proof_of_lemma_4}.
\end{proof}

Then we can construct a tensor network that is similar to Figure~\ref{fig:variance-tensor-network} as follows.
\begin{equation}
  \scalebox{0.6}{\input{hardware-efficient-tn.tikz}}
  \label{eq:hardware-efficient-tn}
\end{equation}
And if we want to compute the variance $\mathbf{Var}(\frac{\partial \braket{H}}{\partial \theta_j})$, we can just simply replace the copy tensor corresponding to $\theta_j$ in the above tensor network with the projection $P_2$.

{By now we have represented the variance as a tensor network. And now we are going to analyze the scaling property of this tensor network when the number of qubits $n$ and the depth of the circuit grow up.}

If we denote \[
  \scalebox{1}{\input{hardware-efficient-tn-block.tikz}},
\] then a layer can be represented as \[
  \scalebox{1}{\input{hardware-efficient-tn-layer.tikz}}.
\] And the whole tensor network in Eq.~\eqref{eq:hardware-efficient-tn} is \[
  \scalebox{1}{\input{hardware-efficient-tn-layers.tikz}}.
\] Hence, the variance will be \begin{equation}
  \scalebox{0.7}{\input{hardware-efficient-variance.tikz}}.
  \label{eq:hardware-efficient-variance}
\end{equation}

We can prove that only two eigenvalues of the matrix $LT$ are 1 and the norms of other eigenvalues are less than 1 (for the complete proof, please refer to Appendix~\ref{sec:analysis_of_examples}). Moreover, the eigenspace corresponding to the eigenvalue 1 is generated with two vectors \begin{equation}
  E_1 = \mathrm{span}\{v_{1,2}\otimes\dots\otimes v_{1,2}, v_{1,3}\otimes\dots\otimes v_{1,3}\}, \quad v_{1,2} = \begin{pmatrix}
  1 \\ 1 \\ 0
  \end{pmatrix},\ v_{1,3} = \begin{pmatrix}
  1 \\ 0 \\ 1
  \end{pmatrix}.
\end{equation}
Hence, $LT^d$ will converge to $P_{E_1}$, the projection to the eigenspace $E_1$, exponentially, as $d \rightarrow \infty$.

If we replace $LT^{L_1}$ and $LT^{L_2}$ with the projection $P_{E_1}$, then the Eq.~\eqref{eq:hardware-efficient-variance} will become \begin{equation}
  \scalebox{0.85}{\input{hardware-efficient-variance-limit.tikz}}.
  \label{eq:hardware-efficient-variance-limit}
\end{equation}
This term is 
\begin{equation}
  4\cdot\frac{1}{4^n}\mathrm{Tr}(H^2),
  \label{eq:he-variance}
\end{equation} which is exponentially small.
{That is, the number of qubits $n$ determines an exponential small limitation of the variance in Eq.~\eqref{eq:he-variance} when the number of layers $L \rightarrow \infty$. And the number of layers $L$ determines how the variance will be close to the limitation.} More precisely, the variance $\mathbf{Var}\left( \frac{\partial \braket{H}}{\partial \theta} \right)$ is exponentially (in $L$) close to the exponential small (in $n$) value Eq.~\eqref{eq:he-variance}. Thus, there exist BPs in the hardware-efficient ansatz.

\begin{theorem}
  The variance of gradients in the hardware-efficient ansatz defined in Eq.~\eqref{eq:hardware-efficient-circuit} vanishes exponentially as the number of qubits $n$ and the number of layers $L$ grow up.
\end{theorem}

Note that the above analysis can be generalized to any hardware-efficient ans\"atze if the entangler connects all of the qubits.

\subsection{Tree tensor network ansatz} 
\label{sub:tree_tensor_network_ansatz}
The tree tensor network is a special kind of tensor network with tree structures. The quantum analog of the tree tensor network was developed in~\cite{grant2018hierarchical}. In~\cite{zhang2020toward}, it was proved that the sum of the variance \[
    \sum^{n}_{j=1} \mathbf{Var}\left( \frac{\partial \braket{H}}{\partial \theta_j} \right)
\] will not vanish exponentially. In this section, we will prove that not only the sum of the variance but also the variance of each parameter vanishes polynomially.

Consider the tree tensor network ansatz with $n$-qubit of the following form.
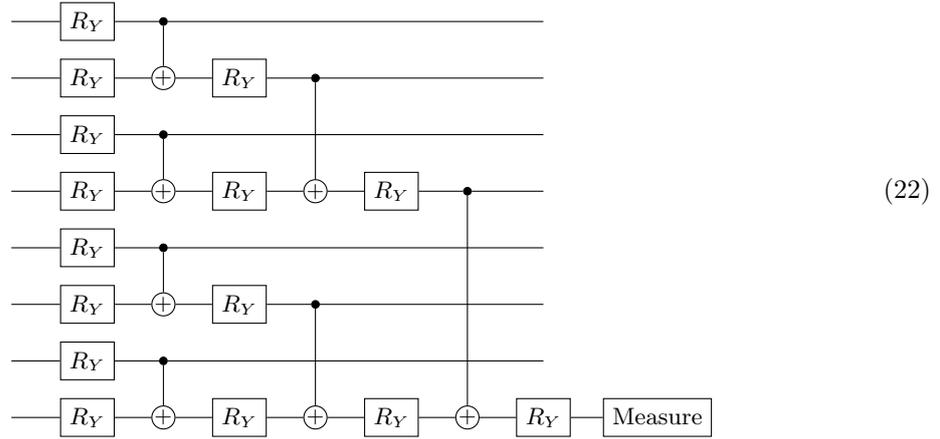
\begin{equation}
  \label{eq-gao42}
      \input{TTN-circuit.tikz}
\end{equation}
To analyze the BP phenomenon of this ansatz, we first use the gate decomposition \[
R_Y(\theta) = R_Z(\frac{\pi}{2}) R_X(\theta) R_Z(-\frac{\pi}{2})
\] to convert the PQC to a ZX-diagram as follows. \[
    \input{TTN-zxd.tikz}
\] The X-spiders with phase ``\dots'' are spiders with parameters. And the ZX-diagram can be rewritten to a graph-like ZX-diagram as follows.
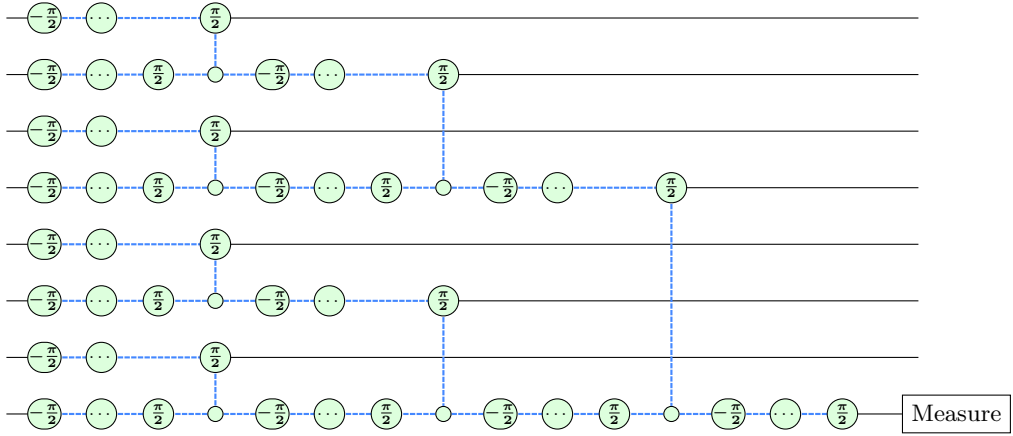
\[
    \input{TTN-zxg.tikz}
\] By using the rewriting rule $(\mathbf{lc})$ in~\cite{duncan2020graph}, we can remove the spiders with phases $\pm\frac{\pi}{2}$.
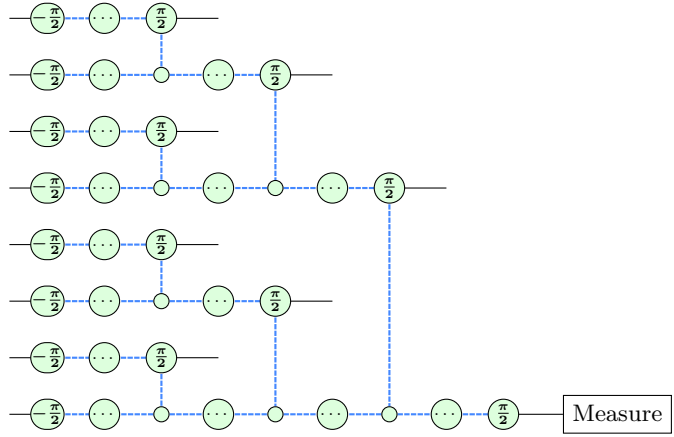
\[
    \input{TTN-zxg2.tikz}
\]

{Now we are going to construct a tensor network from this graph-like ZX-diagram to represent the variance.}
By Eq.~\eqref{eq:gradient-2}, the building block of the variance $\frac{\partial\braket{H}}{\partial\theta_j}^2$ is \[
    \input{TTN-block.tikz}
.\]
We can prove that (for the complete proof, please refer to Appendix~\ref{sec:analysis_of_examples}), after integration over the parameters $\alpha, \beta, \gamma$, the building block will become
\begin{equation}
    \scalebox{0.7}{\input{integration-TTN.tikz}}
    \label{eq:integration-TTN}
\end{equation} Here, $T_{ \mathrm{TTN}}$ is a $3\times 3\times 3$ tensor defined as follows. \begin{equation*}
T_{ \mathrm{TTN}}[1,\cdot,\cdot] = \frac{1}{16}\begin{pmatrix}
    1 & 0 & 1 \\
    0 & 1 & 0 \\
    1 & 0 & 1
\end{pmatrix}, \quad
T_{ \mathrm{TTN}}[2,\cdot,\cdot] = \frac{1}{16}\begin{pmatrix}
    1 & 0 & -1 \\
    0 & 1 & 0 \\
    -1 & 0 & 1
\end{pmatrix}, \quad
T_{ \mathrm{TTN}}[3,\cdot,\cdot] = \frac{1}{16}\begin{pmatrix}
    1 & 0 & 1 \\
    0 & 1 & 0 \\
    1 & 0 & 1
\end{pmatrix}
.\end{equation*}

Hence, the variance of $\frac{\partial \braket{H}}{\partial \theta_j}$ can be obtained by replacing one of the copy tensors with the projection $P_2$ in the following tensor network.
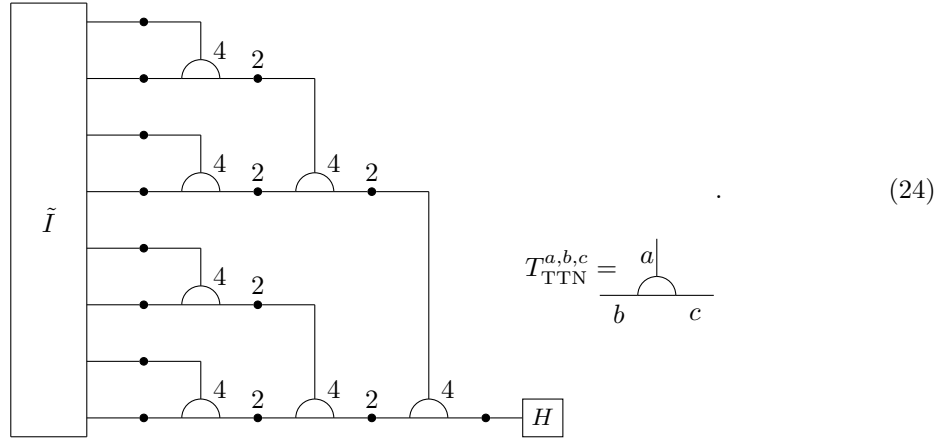
\begin{equation}
    \input{TTN-variance.tikz}
.\end{equation}

{Now let us analyze the scaling property of this tensor network.}
Since the Hamiltonian $H$ is a 1-qubit Hermitian operator, it can be expressed as \[
    H = k_0 I + k_1 X + k_2 Y + k_3 Z,\quad k_j \in \mathbb{R}.
\] Then \[
    \tilde{H} = 2k_0^2\begin{pmatrix}
        1 \\ 0 \\ 1
    \end{pmatrix} + 2(k_1^2+k_3^2) \begin{pmatrix}
        0 \\ 1 \\ 0
    \end{pmatrix} + 2k_2^2 \begin{pmatrix}
        1 \\ 0 \\ -1
    \end{pmatrix}.
\] We denote \[
    v_2 = \begin{pmatrix} 0 \\ 1 \\ 0 \end{pmatrix},\quad
    v_{1,3} = \begin{pmatrix} 1 \\ 0 \\ 1 \end{pmatrix}, \quad
    v_{1,3}^- = \begin{pmatrix} 1 \\ 0 \\ -1 \end{pmatrix}.
\]
Note that the building block of this tensor network is $8\cdot T_{\mathrm{TTN}}$. By the definition of $T_{\mathrm{TTN}}$, we have \begin{equation}
    \scalebox{0.8}{\input{TTN-analysis-1.tikz}}
    \label{eq:TTN-analysis-1}
\end{equation}
With the above equations, we can compute the variance simply.
For example, consider the following variance. \[
    \scalebox{0.8}{\input{TTN-example-1.tikz}}.
\] It is a linear function of $\tilde{H}$. Hence, we can analyze each term of $\tilde{H}$ individually.

Since $P_2v_{1,3}=0$, the first term $2k_0^2v_{1,3}$ in $\tilde{H}$ will become 0.

Now let us consider the second term $2(k_1^2+k_3^2)v_2$. With Eq.~\eqref{eq:TTN-analysis-1}, we can expand the variance as follows.  \[
    \scalebox{0.8}{\input{TTN-example-2.tikz}}
\] Expanding it recursively, the variance can be represented as a summation of terms of the following form, 
\begin{equation}
\tilde{I}(u_1,\dots,u_n) = \scalebox{1}{\input{TTN-example-3.tikz}}
\label{eq:tilde-I-u}
\end{equation}
And by the definition of $\tilde{I}$, each of the term $\tilde{I}(u_1,\dots,u_n) \geq 0$. Hence, we can obtain a lower bound, \[
    \scalebox{1}{\input{TTN-example-4.tikz}}
\]

Similarly, we have a lower bound for the term $v_{1,3}^-$.

For the general case of $n$-qubit, we can prove that it has a lower bound.
\begin{restatable}{theorem}{TTNLowerBound}
For tree tensor network ansatz shown in~\eqref{eq-gao42},
if \[
        H = k_0I+k_1X+k_2Y+k_3Z,
    \] then we have \begin{equation}
        \label{eq:TTN_lower_bound}
        \mathbf{Var}\left( \frac{\partial \braket{H}}{\partial \theta} \right) \geq \frac{k_1^2+k_3^2}{n^2}\tilde{I}(u_1,\dots,u_n) + \frac{k_2^2}{n^2}\tilde{I}(w_1,\dots,w_n),
    \end{equation} for some $u_j,w_j\in \{ v_{1,3}, v_2, v_{1,3}^- \} $. Here $\tilde{I}(u_1,\dots,u_n)$ is defined in Eq.~\eqref{eq:tilde-I-u}. And $\tilde{I}$ is a $3^n$ dimensional tensor which only depends on the input state. If the input state is $\rho$, then $\tilde{I}$ is defined as follows. \[
        \scalebox{1}{\input{TTN-input.tikz}}
    \]
\end{restatable}
\begin{proof}
    See Appendix~\ref{appendix:tree_tensor_network_ansatz}.
\end{proof}
Note that $\tilde{I}(u_1,\dots,u_n)$ only depends on the input state. If \[
    \tilde{I}(u_1,\dots,u_n)\in \Omega(\frac{1}{\mathrm{poly}(n)}) \text{ or } \tilde{I}(w_1,\dots,w_n)\in \Omega(\frac{1}{\mathrm{poly}(n)}),
\] there exist no BP in the tree tensor network ansatz.


\subsection{QCNN}%
\label{sub:qcnn}

QCNN was developed in~\cite{cong2019quantum}. It was proved that there exists no BP in the QCNN ansatz if the subblocks form unitary 2-design~\cite{pesah2020absence}. In this section, we will use the ZX-calculus to analyze the BP phenomenon in a QCNN ansatz without the assumption of unitary 2-design.

Consider a QCNN ansatz as follows.
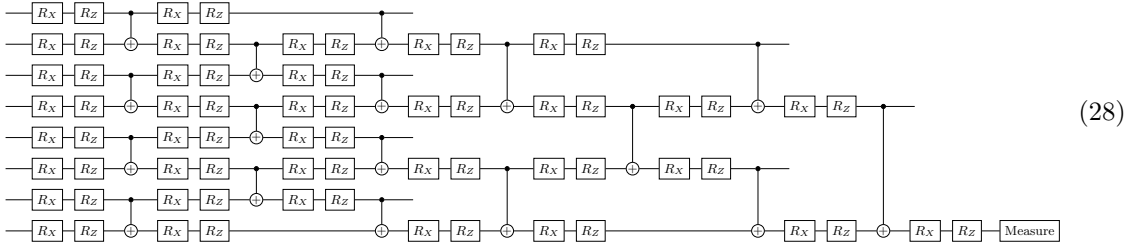
\begin{equation}
\scalebox{0.55}{\input{QCNN-circuit.tikz}}
\label{eq-gao43}
\end{equation} 
It can be represented as the following ZX-diagram.
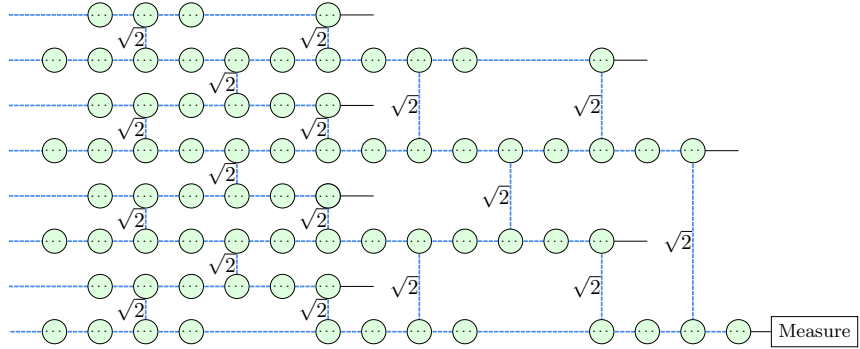
\[
    \scalebox{0.8}{\input{QCNN-zxg.tikz}}
\] where the Z-spiders with ``\dots'' are parameterized. Note that this is a graph-like ZX-diagram whose spiders are all parameterized. Hence, by using Lemma~\ref{thm:integration-hadamard}, the variance can be obtained by replacing one of the copy tensors with the projection $P_2$ in the following tensor network.
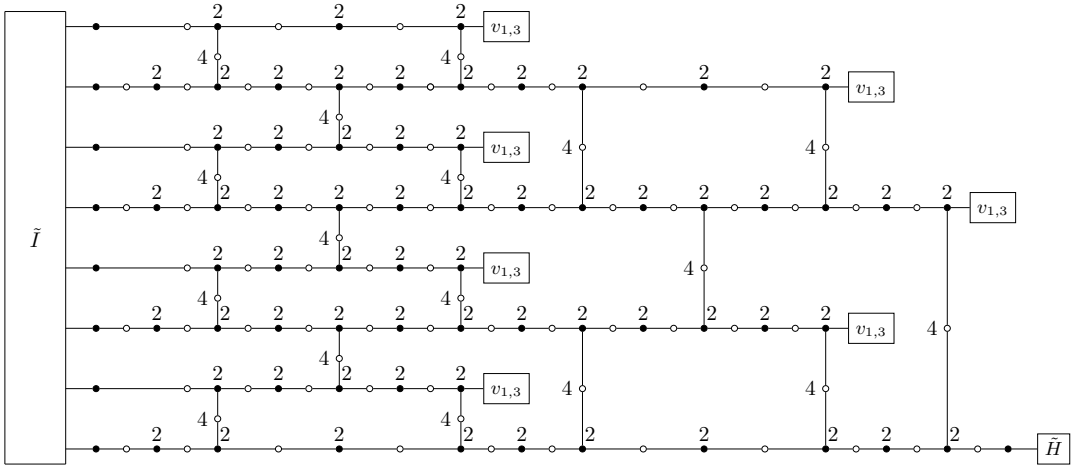
\[
    \scalebox{0.8}{\input{QCNN-variance.tikz}}
\]

{Now we are going to analyze the scaling property of this tensor network.}
If we denote \[
    \input{QCNN-block.tikz},
\] then we have the following equations 
\begin{equation}
    \scalebox{0.7}{\input{QCNN-analysis.tikz}}
    \label{eq:QCNN-analysis}
\end{equation} By using Eq.~\eqref{eq:QCNN-analysis}, we can expand the variance as a sum of terms of the following form. \[
\scalebox{1}{\input{QCNN-term.tikz}}
\] Each of these terms is non-negative. Hence, similar to the analysis of tree tensor network ansatz,  we can prove that the variance of gradients in the QCNN ansatz has a lower bound, since the QCNN ansatz is of $O(\log(n))$-depth.
\begin{restatable}{theorem}{QCNNLowerBound}
    For the QCNN ansatz shown in~\eqref{eq-gao43},
    if \[
        H = k_0I+k_1X+k_2Y+k_3Z,
    \] then we have \begin{equation}
    \mathbf{Var}\left( \frac{\partial \braket{H}}{\partial \theta} \right) \geq \frac{k_2^2+k_3^2}{n^9}\tilde{I}(u_1,\dots,u_n) + \frac{k_1^2}{n^9}\tilde{I}(w_1,\dots,w_n),
    \end{equation} for some $u_j,w_j\in \{v_{1,3},v_2,v_{1,3}^-\}$. Here $\tilde{I}$ is a $3^n$ dimensional tensor which only depends on the input state $\rho$.
    \[
        \scalebox{1}{\input{QCNN-input.tikz}}
    \]
\end{restatable}
\begin{proof}
    See Appendix~\ref{appendix:qcnn}.
\end{proof}
Hence, if provided \[
    \tilde{I}(u_1,\dots,u_n) \in {\Omega}( \frac{1}{\mathrm{poly}(n)} ) \text{ or } \tilde{I}(w_1,\dots,w_n) \in {\Omega}( \frac{1}{\mathrm{poly}(n)} ),
\] then there exists no BP in the QCNN ansatz.

\subsection{MPS-inspired ansatz}%
\label{sub:mps_inspired_ansatz}

The matrix product state (MPS) is a special structure of tensor networks, which  is widely used in quantum physics and machine learning~\cite{doi:10.1080/14789940801912366,PhysRevX.8.031012}.
There are also PQCs with a similar structure as MPS, and we call it MPS-inspired ansatz. It has been shown that MPS-inspired ansatz can be implemented efficiently in quantum computers with a small number of  qubits~\cite{PhysRevResearch.1.023025}. We will analyze the BP phenomenon in MPS-inspired ansatz in this section.

Let us consider the following MPS-inspired ansatz
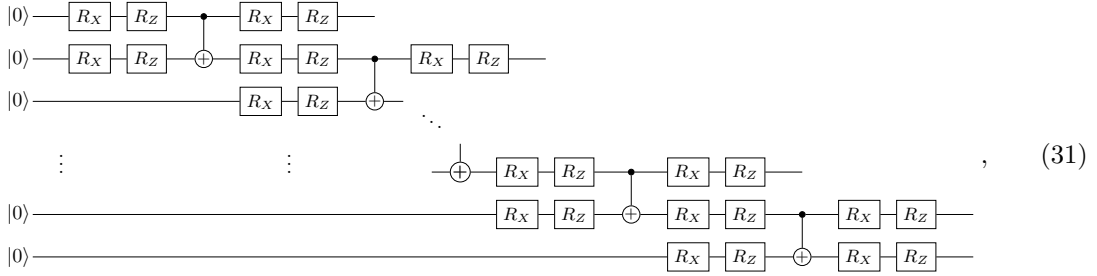
\begin{equation}
\label{eq-gao44}
    \scalebox{0.75}{\input{MPS-circuit.tikz}},
\end{equation}
and the Hamiltonian $H=I\otimes I\dots\otimes I\otimes X$. We will prove that the variance $\mathbf{Var}\left( \frac{\partial \braket{H}}{\partial \theta_1}  \right) $ is exponentially small. Here $\theta_1$ is the parameter of the first $R_X$ gate applying on the first qubit.

Firstly, we convert the PQC into a ZX-diagram as follows.
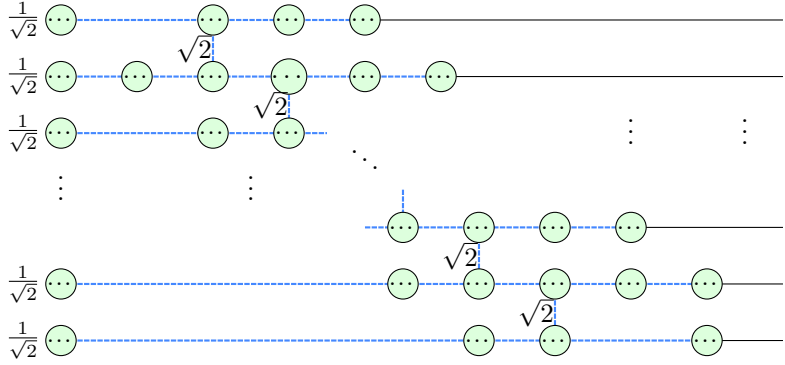
\[
    \scalebox{1}{\input{MPS-zxg.tikz}}
\] This is a graph-like ZX-diagram whose spiders are all parameterized. We can use Lemma~\ref{thm:integration-hadamard} to represent the variance as the following tensor network.
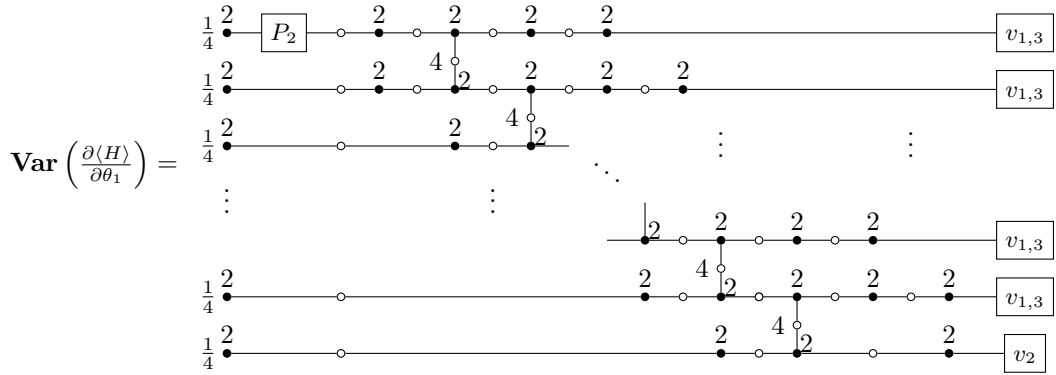
\[
    \scalebox{1}{\input{MPS-variance.tikz}}
\] By using \[
    2Mv_2 = \frac{1}{2}\left( v_2+v_{1,3}^- \right) ,\quad 2Mv_{1,3} = v_{1,3}
,\] and \[
\scalebox{0.7}{\input{MPS-analysis.tikz}}
\] we can simplify the variance as \begin{equation}
    \scalebox{1}{\input{MPS-variance-3.tikz}}
\end{equation}
It is exponential in the number of qubits $n$. Hence, there exist BPs in the MPS-inspired ansatz.


\section{Discussion}%
\label{sec:discussion}

We developed powerful techniques to analyze the BP phenomenon for quantum neural networks training with the ZX-calculus. The quantum neural networks under consideration are PQCs under certain reasonable assumptions 
and the cost function is the expectation $\braket{H}$ of the PQC with respect to a given Hamiltonian $H$. The basic idea of the method is to represent the PQC, the cost function $\braket{H}$, and the gradients $\frac{\partial \braket{H}}{\partial \theta_j} $ as ZX-diagrams.
And then computing the expectation and the variance of the gradient of $\braket{H}$
becomes computing the integration of certain ZX-diagrams.
We show that these integrations are sums of ZX-diagrams that can be
computed explicitly in many cases.
{As future works, it would be desirable to  use the completeness of ZX-calculus to represent these sums by ZX-diagrams, or more generally, diagrams in other complete graphical calculi.}

In principle, these techniques can be used to any given ansatz under Assumption~\ref{assumption-PQC}.
We remark that these techniques can be used to analyze the BP phenomenon for PQCs which contain $t$-design sub-blocks, for example, the PQCs considered in~\cite{cerezo2020cost,pesah2020absence}. Because the $t$-design sub-blocks can be replaced with concrete $t$-design PQCs and then the techniques proposed in this paper can be applied.
{
Techniques introduced in this paper can be used in more cases
including circuits with global cost functions and circuits of any depth.
To analyze a PQC with a global cost function, one can represent the global Hamiltonian as a ZX-diagram and then the method introduced in this paper can be used.
As shown in Section~\ref{sub:hardware_efficient_ansatz}, the BP phenomenon for the hardware-efficient ansatz
has been analyzed when the depth is $O(\mathrm{poly}(n))$.
}
In conclusion, we extend the BP theorem from unitary 2-design circuits to any parameterized quantum circuits under Assumption~\ref{assumption-PQC}.

Using the techniques proposed in this paper, we analyzed four kinds of ans\"atze,
including the hardware-efficient ansatz, the tree tensor network ansatz, the QCNN ansatz, and the MPS-inspired ansatz. It is shown that there exist BPs in the hardware-efficient-ansatz and the MPS-inspired ansatz, while there exists no BP in the tree tensor network ansatz and the QCNN ansatz.

\vskip20pt
{\bf Acknowledgment.}
This work is partially supported by a NSFC grant No.11688101 and by a NKRDP grant No.2018YFA0306702.
We want to thank the anonymous referees for their valuable suggestions. 

\bibliographystyle{unsrtnat}
\bibliography{zx-bp}

\appendix
\section{A simple example}%
\label{appendix:a_simple_example}

In this section, we will use a simple example  considered in Section~\ref{sub:representing_gradients_as_zx_diagrams} to illustrate
the techniques proposed in this paper.
Recall that the PQC  considered is
\[
    \scalebox{1}{\input{para-shift-circuit.tikz}}
\]
and the gradient with respect to $\theta_1$ can be represented as
\[
    \scalebox{1}{\input{para-shift-gradient.tikz}}.
\]
By Lemma~\ref{thmt@@intone}, the expectation of this gradient is zero. The variance of the gradient is the following integration.
\[
    \scalebox{1}{\input{para-shift-variance.tikz}}.
\]
Now we choose the Hamiltonian $H=X\otimes X$. Then by Lemma~\ref{thmt@@inttwo} and Theorem~\ref{thm:var_theta_j}, we have
\[
    \scalebox{0.9}{\input{para-shift-variance-2.tikz}}.
\]
Using Lemma~\ref{thmt@@inthad} recursively, we can break the large ZX-diagram into small parts as follows.
\[
    \scalebox{0.8}{\input{para-shift-variance-3.tikz}}
\]
We denote
\[
    \scalebox{1}{\input{para-shift-variance-4.tikz}}
\]
Then the variance can be represented as the following tensor network.
\[
    \scalebox{1}{\input{para-shift-variance-5.tikz}}
\] By contracting this tensor network, we finally obtain
\[
    \mathbf{Var}\left( \frac{\partial \braket{H}}{\partial \theta_1} \right) = \frac{1}{16}\cdot \frac{3}{4}.
\]

\section{Proof of theorem 2}%
\label{sec:proof_of_theorem_2}
\zxgrad*
\begin{proof}
    Consider the spiders corresponding to $\theta_j$. We expand the spiders as follows.
    \[
        \scalebox{0.8}{\input{theorem2-1.tikz}}
    \] Taking the partial derivative of $\theta_j$ on the two sides, we obtain
    \[
        \scalebox{0.8}{\input{theorem2-2.tikz}}
    \]
\end{proof}

\section{Proof of lemma 1}%
\label{sec:proof_of_lemma_1}
\intone*
\begin{proof}
    We expand the spiders as follows.
    \[
        \scalebox{0.8}{\input{lemma1-1.tikz}}
    \] By \[
    \int_{-\pi }^{{\pi}} {e^{ik \alpha}} \: d{\alpha} = 0,\quad k=\pm 1,
    \] we have \[
        \scalebox{0.8}{\input{lemma1-2.tikz}}
    \]
\end{proof}

\section{Proof of lemma 2}%
\label{sec:proof_of_lemma_2}
\inttwo*
\begin{proof}
    We expand each spider on the left-hand side of the equation as follows.
    \[
        \scalebox{0.5}{\input{lemma2-1.tikz}}
    \]
    Since \[
        \int_{0}^{2\pi} {e^{ik \alpha}} \: d{\alpha} = 0,\quad k=\pm 1, \pm 2,
    \] we integrate over $\alpha$ on each side and obtain \[
        \scalebox{0.5}{\input{lemma2-2.tikz}}
    \]
\end{proof}

\section{Proof of lemma 3}%
\label{sec:proof_of_lemma_3}
\inthad*
\begin{proof}
    By Lemma~\ref{thmt@@inttwo}, we have
    \[
        \scalebox{0.7}{\input{lemma3-1.tikz}}.
    \] Hence, it is sufficient to prove \[
    \scalebox{0.8}{\input{lemma3-2.tikz}}
    \] for $a_j,a_j\in \{T_1,T_2,T_3\} $.

    For $a_j=a_k=T_1$, we have \begin{equation}
        \label{eq:lemma3-T1T1}
        \scalebox{0.8}{\input{lemma3-T1T1.tikz}}
    \end{equation} For $(a_j,a_k)\notin \{(T_2,T_3),(T_3,T_2)\} $, the proof is almost the same as that of Eq.~\eqref{eq:lemma3-T1T1}.

    Now, let us consider the case when $(a_j,a_k) = (T_2,T_3)$. We can use the rules in Figure~\ref{fig:zx_rules_scalar} as follows. \[
        \scalebox{0.8}{\input{lemma3-T2T3.tikz}}
    \]
\end{proof}

\section{Proof of lemma 4}%
\label{sec:proof_of_lemma_4}
\intcnot*
\begin{proof}
    By Lemma~\ref{thmt@@inttwo}, it is sufficient to prove that \[
        \scalebox{0.8}{\input{lemma4-1.tikz}}
    \] for $a_1,a_2,a_3\in \{T_1,T_2,T_3\}$.

    We first consider the case when $a_1=T_1$. Since $a_1=T_1$, we have \[
        \scalebox{0.8}{\input{lemma4-T1-1.tikz}}
    \] Now, if $a_2=T_1$, then we have \[
        \scalebox{0.7}{\input{lemma4-T1-2.tikz}}
    \] Hence, \[
        \scalebox{0.7}{\input{lemma4-T1-3.tikz}}
    \] That is $ET[1,1,\cdot] = \frac{1}{8}\begin{pmatrix} 1 & 0 & 0 \end{pmatrix}$.

    If $a_2=T_2$, then \[
        \scalebox{0.8}{\input{lemma4-T1-4.tikz}}
    \] Hence, \[
    \scalebox{0.8}{\input{lemma4-T1-5.tikz}}
    \] That is $ET[1,2,\cdot] = \frac{1}{8}\begin{pmatrix} 0 & 1 & 0 \end{pmatrix}$.

    If $a_2=T_3$, then \[
        \scalebox{0.8}{\input{lemma4-T1-6.tikz}}
    \] Hence. \[
        \scalebox{0.8}{\input{lemma4-T1-7.tikz}}
        \] That is $ET[1,3,\cdot] = \frac{1}{8}\begin{pmatrix} 0 & 0 & 1 \end{pmatrix}$.

    By now, we have proved that \begin{equation}
        ET[1,\cdot,\cdot] = \frac{1}{8}\begin{pmatrix} 1 & 0 & 0 \\ 0 & 1 & 0 \\ 0 & 0 & 1 \end{pmatrix}.
    \end{equation}

    Now, let us consider the case when $a_1=T_2$. Since $a_1=T_2$, we have \[
        \scalebox{0.8}{\input{lemma4-T2-1.tikz}}
    \] Now, if $a_2=T_1$, then we have \[
        \scalebox{0.7}{\input{lemma4-T2-2.tikz}}
    \] Hence, \[
        \scalebox{0.7}{\input{lemma4-T2-3.tikz}}
    \] That is $ET[1,1,\cdot] = \frac{1}{8}\begin{pmatrix} 0 & 1 & 0 \end{pmatrix}$.

    If $a_2=T_2$, then \[
        \scalebox{0.8}{\input{lemma4-T2-4.tikz}}
    \] Hence, \[
    \scalebox{0.8}{\input{lemma4-T2-5.tikz}}
    \] That is $ET[1,2,\cdot] = \frac{1}{8}\begin{pmatrix} 1 & 0 & 0 \end{pmatrix}$.

    If $a_2=T_3$, then \[
        \scalebox{0.8}{\input{lemma4-T2-6.tikz}}
    \] Hence. \[
        \scalebox{0.8}{\input{lemma4-T2-7.tikz}}
        \] That is $ET[1,3,\cdot] = \frac{1}{8}\begin{pmatrix} 0 & 0 & 0 \end{pmatrix}$.

    By now, we have proved that \begin{equation}
        ET[2,\cdot,\cdot] = \frac{1}{8}\begin{pmatrix} 0 & 1 & 0 \\ 1 & 0 & 0 \\ 0 & 0 & 0 \end{pmatrix}.
    \end{equation}

    Let us consider the case when $a_1=T_3$.

    When $a_1=T_3$, we have \[
        \scalebox{0.8}{\input{lemma4-T3-1.tikz}}
    \] According to the symmetry, we can use the result of the case when $a_1=T_2$. And we obtain \begin{equation}
    ET[3,\cdot,\cdot] = \frac{1}{8}\begin{pmatrix} 0 & 0 & 1 \\ 0 & 0 & 0 \\ 1 & 0 & 0 \end{pmatrix}.
    \end{equation}
\end{proof}

\section{Analysis of PQCs in section 4}%
\label{sec:analysis_of_examples}

\subsection{Hardware-efficient ansatz}%
\label{appendix:hardware_efficient_ansatz}

We will prove some properties of $LT$, which are used in the analysis in Section~\ref{sub:hardware_efficient_ansatz}.
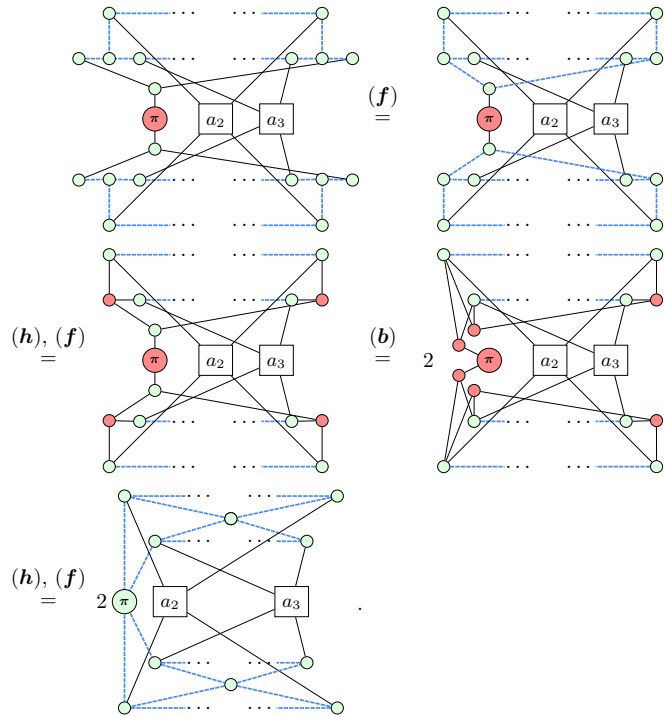
\begin{theorem}
    Suppose that $\lambda_1,\dots,\lambda_{3^n}$ are eigenvalues of $LT$. And \[
        \left| \lambda_1 \right| \geq \left| \lambda_2 \right| \geq \dots \geq \left| \lambda_{3^n} \right| .
    \] Then we have \[
    \lambda_1= \lambda_2 = 1,\quad |\lambda_j| < 1 \text{ for } j>2.
    \]
\end{theorem}
\begin{proof}
    By definition of $EM$, we can compute that \begin{equation}
        \label{eq:matrix_EM}
        EM =
        \begin{pmatrix}
            \frac{3}{4} & \frac{1}{4} & \frac{1}{4} & & & & & & \\
            \frac{1}{4} & \frac{3}{4} & -\frac{1}{4} & & & & & & \\
            \frac{1}{4} & -\frac{1}{4} & \frac{3}{4} & & & & & & \\
            & & & \frac{1}{4} & \frac{3}{4} & 0 & & & \\
            & & & \frac{3}{4} & \frac{1}{4} & 0 & & & \\
            & & & -\frac{1}{4} & \frac{1}{4} & 0 & & & \\
            & & & & & & \frac{1}{4} & 0 & \frac{3}{4} \\
            & & & & & & -\frac{1}{4} & 0 & \frac{1}{4} \\
            & & & & & & \frac{3}{4} & 0 & \frac{1}{4}
        \end{pmatrix}.
    \end{equation}
    By computation, $EM$ can be diagonalized. Four of its eigenvalues are 1 and other eigenvalues are in the interval $(-1,1)$. Moreover, the eigenspace of the eigenvalue 1 is
    \begin{equation}
        \label{eq:eigen_EM}
        \mathrm{span}\left\{ v_1\otimes v_{1,2}, v_1\otimes v_{1,3}, v_{1,2}\otimes v_{1,2}, v_{1,3}\otimes v_{1,3} \right\},
    \end{equation} where \[
    v_1 = \begin{pmatrix} 1 \\ 0 \\ 0 \end{pmatrix} , v_{1,2} = \begin{pmatrix} 1 \\ 1 \\ 0 \end{pmatrix} , v_{1,3} = \begin{pmatrix} 1 \\ 0 \\ 1 \end{pmatrix} .
    \]

    $LT$ is an operator on the tensor product of $n$  $\mathbb{R}^3$. We denote the operator $EM$ on the $i$-th and $j$-th $\mathbb{R}^3$ as $EM_{i,j}$. Then the eigenspace of $LT$ corresponding to the eigenvalue 1 is the intersection of the eigenspaces corresponding to the eigenvalue 1 of \[
        EM_{1,2},\ EM_{2,3},\ \dots, EM_{n-1,n},\ EM_{n,1}.
    \] Hence, by Eq.~\eqref{eq:eigen_EM}, we have \[
        E_1 = \mathrm{span}\left\{ v_{1,2}\otimes\dots\otimes v_{1,2}, v_{1,3}\otimes\dots\otimes v_{1,3} \right\}.
    \]
\end{proof}

\subsection{Tree tensor network ansatz}%
\label{appendix:tree_tensor_network_ansatz}

In Section~\ref{sub:tree_tensor_network_ansatz}, we used Eq.~\eqref{eq:integration-TTN}. Here we will prove this equation.
\begin{proof}[Proof of Eq.~\eqref{eq:integration-TTN}]
    By Lemma~\ref{thmt@@inttwo}, it suffices to prove that \[
        \scalebox{0.8}{\input{TTN-proof-1.tikz}}.
    \]
    Similar to the proof of Lemma~\ref{thmt@@intcnot}, we first consider the case when $a=T_1$.
    When $a=T_1$, we have \[
        \scalebox{0.7}{\input{TTN-proof-T1-1.tikz}}.
    \] Now, $a$ is disconnected with $b$ and $c$. Hence, we can consider $b$ and $c$ separately. If $b=T_1$, then \[
        \scalebox{0.7}{\input{TTN-proof-T1-2.tikz}}
    \] Hence, \[
        \scalebox{0.8}{\input{TTN-proof-T1-3.tikz}}
    \] That it \[
        T_{\mathrm{TTN}}[1,1,\cdot] = \frac{1}{16}\begin{pmatrix} 1 & 0 & 1 \end{pmatrix}.
    \] If $b=T_2$, then \[
        \scalebox{0.6}{\input{TTN-proof-T1-4.tikz}}
    \] Hence, \[
        \scalebox{0.8}{\input{TTN-proof-T1-5.tikz}}
    \] That it \[
        T_{\mathrm{TTN}}[1,2,\cdot] = \frac{1}{16}\begin{pmatrix} 0 & 1 & 0 \end{pmatrix}.
    \] If $b=T_3$, then \[
        \scalebox{0.7}{\input{TTN-proof-T1-6.tikz}}
    \] Hence, \[
        \scalebox{0.8}{\input{TTN-proof-T1-3.tikz}}
    \] That it \[
        T_{\mathrm{TTN}}[1,3,\cdot] = \frac{1}{16}\begin{pmatrix} 1 & 0 & 1 \end{pmatrix}.
    \] By now, we have proved that \begin{equation}
        T_{\mathrm{TTN}}[1,\cdot,\cdot] = \frac{1}{16}\begin{pmatrix} 1 & 0 & 1 \\ 0 & 1 & 0 \\ 1 & 0 & 1 \end{pmatrix} .
    \end{equation}

    Now, let us consider the case when $a=T_2$.
    When $a=T_2$, we have \[
        \scalebox{0.7}{\input{TTN-proof-T2-1.tikz}}.
    \] Now, $a$ is disconnected with $b$ and $c$. Hence, we can consider $b$ and $c$ individually. If $b=T_1$, then \[
        \scalebox{0.7}{\input{TTN-proof-T2-2.tikz}}
    \] Hence, \[
        \scalebox{0.8}{\input{TTN-proof-T2-3.tikz}}
    \] That it \[
        T_{\mathrm{TTN}}[2,1,\cdot] = \frac{1}{16}\begin{pmatrix} 1 & 0 & -1 \end{pmatrix}.
    \] If $b=T_2$, then \[
        \scalebox{0.6}{\input{TTN-proof-T2-4.tikz}}
    \] Hence, \[
        \scalebox{0.8}{\input{TTN-proof-T2-5.tikz}}
    \] That it \[
        T_{\mathrm{TTN}}[2,2,\cdot] = \frac{1}{16}\begin{pmatrix} 0 & 1 & 0 \end{pmatrix}. ow
    \] If $b=T_3$, then \[
        \scalebox{0.7}{\input{TTN-proof-T2-6.tikz}}
    \] Hence, \[
        \scalebox{0.8}{\input{TTN-proof-T2-7.tikz}}
    \] That it \[
        T_{\mathrm{TTN}}[2,3,\cdot] = \frac{1}{16}\begin{pmatrix} -1 & 0 & 1 \end{pmatrix}.
    \] By now, we have proved that \begin{equation}
        T_{\mathrm{TTN}}[2,\cdot,\cdot] = \frac{1}{16}\begin{pmatrix} 1 & 0 & -1 \\ 0 & 1 & 0 \\ -1 & 0 & 1 \end{pmatrix} .
    \end{equation}

    Now, let us consider the case when $a=T_3$.
    When $a=T_3$, we have \[
        \scalebox{0.7}{\input{TTN-proof-T3-1.tikz}}.
    \] Now, $a$ is disconnected with $b$ and $c$. And the parts of $b$ and $c$ are the   same as that of $a=T_1$. Hence, we have
    \begin{equation}
        T_{\mathrm{TTN}}[3,\cdot,\cdot] = \frac{1}{16}\begin{pmatrix} 1 & 0 & 1 \\ 0 & 1 & 0 \\ 1 & 0 & 1 \end{pmatrix} .
    \end{equation}
\end{proof}

Now we are going to prove that there is a lower bound for the variance of gradients in the tree tensor network ansatz.

\TTNLowerBound*
\begin{proof}
    We first prove that
    \[
         \tilde{I}(u_1,\dots,u_n) \geq 0,\quad u_j\in \left\{ v_2,v_{1,3}, v_{1,3}^- \right\},
    \] where
    \[
        v_2 = \begin{pmatrix} 0 \\ 1 \\ 0 \end{pmatrix} , \quad v_{1,3} = \begin{pmatrix} 1 \\ 0 \\ 1 \end{pmatrix} ,\quad v_{1,3}^-=\begin{pmatrix} 1 \\ 0 \\ -1 \end{pmatrix}.
    \]
    Note that if $u_1$ is $v_{1,3}$ or $v_{1,3}^-$, then
    \[
        \scalebox{1}{\input{TTN-input-v13-1.tikz}}.
    \] By
    \[
        \scalebox{1}{\input{TTN-input-v13-2.tikz}},
    \] we can expand $\tilde{I}(u_1,\dots,u_n)$ as a sum of squares. Thus, we have proved that \[
        \tilde{I}(u_1,\dots,u_n)\geq 0.
    \]

    Now let us consider the lower bound of $\mathbf{Var}\left( \frac{\partial \braket{H}}{\partial \theta} \right)$. By the graph-like ZX-diagram just obtained, we have $\tilde{H} = \begin{pmatrix} \tilde{H_1} \\ \tilde{H_2} \\ \tilde{H_3} \end{pmatrix}$ which is defined as follows.
    \[
        \scalebox{1}{\input{TTN-tilde-H.tikz}}
    \]
    Suppose that $H = k_0I+k_1X+k_2Y+k_3Z$. Then we can obtain
    \begin{equation}
        \tilde{H} = 2k_0^2 v_{1,3} + 2(k_1^2+k_3^2) v_2 + 2k_2^2 v_{1,3}^-.
    \end{equation}
    And by the definition of $T_{\mathrm{TTN}}$, we can obtain Eq.~\eqref{eq:TTN-analysis-1}. Hence, we can expand the variance as
    \[
        \mathbf{Var}\left( \frac{\partial \braket{H}}{\partial \theta} \right) = \sum_{u_j\in \{v_2,v_{1,3},v_{1,3}^-\} } a(u_1,\dots,u_n)\cdot \tilde{I}(u_1,\dots,u_n),
    \] where $a(u_1,\dots,u_n)$ is a non-negative number. We will prove that there exists one choise of $(u_1,\dots,u_n)$ such that \[
        a(u_1,\dots,u_n) \in \Omega(\frac{1}{\mathrm{poly}(n)}).
    \]

    When we use Eq.~\eqref{eq:TTN-analysis-1} to expand the variance, the only case that will cause a coefficient $<1$ is the case when we use the second equation in Eq.~\eqref{eq:TTN-analysis-1}. Hence, the coefficients $a(u_1,\dots,u_n)$ only depend on the number of times that we use the second equation. And it depends on the location of $\theta$. The worst case is that $\theta$ is the parameter of the first gate applying to the first qubit. In this case, we need to use the second equation in Eq.~\eqref{eq:TTN-analysis-1} for $2\log(n)-1$ times. That means we have a lower bound for the variance
    \[
        \mathbf{Var}\left( \frac{\partial \braket{H}}{\partial \theta} \right) \geq \frac{k_1^2+k_3^2}{n^2}\tilde{I}(u_1,\dots,u_n) + \frac{k_2^2}{n^2}\tilde{I}(w_1,\dots,w_n),
    \]
    for some $u_j,w_j\in \{ v_{1,3},v_2,v_{1,3}^- \}$.

Note that $I(u_1,\dots,u_n)$ only depends on the input state $\rho$. If  \[
    \tilde{I}(u_1,\dots,u_n)\in \Omega(\frac{1}{\mathrm{poly}(n)}) \text{ or } \tilde{I}(w_1,\dots,w_n)\in \Omega(\frac{1}{\mathrm{poly}(n)}),
    \]
    then there exists no BP in the tree tensor network ansatz.
\end{proof}

\subsection{QCNN}%
\label{appendix:qcnn}

\QCNNLowerBound*
\begin{proof}
    The proof is similar to that of the tree tensor network ansatz (Theorem~\ref{thmt@@TTNLowerBound}), so
we will only give a sketch of the proof. We also have
    \[
        \tilde{I}(u_1,\dots,u_n) \geq 0, \quad u_j\in \{v_{1,3}, v_2, v_{1,3}^-\} .
    \]
From the graph-like ZX-diagram, we have \[
        \scalebox{1}{\input{QCNN-tilde-H.tikz}}
    \] Hence,   \[
        \tilde{H} = 2k_0^2v_{1,3}+2(k_1^2+k_2^2)v_2+2k_3^2v_{1,3}^-.
    \]We will analyze each term of $\tilde{H}$ in the variance.

    Using Eq.~\eqref{eq:QCNN-analysis}, the term $2k_0^2v_{1,3}$ will become 0 by 
        $P_2v_{1,3}=0.$
 The terms $2(k_1^2+k_2^2)v_2$ and $2k_3^2v_{1,3}^-$ will generate terms containing $v_2$ after expanding using Eq.~\eqref{eq:QCNN-analysis}. And each time after generating terms containing $v_2$, a coefficient $\geq \frac{1}{8}$ will be multiplied to the variance. Hence, if we want to bring $v_2$ to $P_2$, we need to generate terms containing $v_2$ for $l$ times, where $l$ is a path from the location of $v_2$ to the location of $P_2$.

    By the structure of the QCNN ansatz, we have \[
        l\geq 3\log(n).
    \] It will generated a coefficient  $\geq \frac{1}{8^{3\log(n)}}$.

    Hence we have \[
        \mathbf{Var}\left( \frac{\partial \braket{H}}{\partial \theta} \right) \geq \frac{k_2^2+k_3^2}{n^9}\tilde{I}(u_1,\dots,u_n) + \frac{k_1^2}{n^9}\tilde{I}(w_1,\dots,w_n),
    \] for some $u_j,w_j\in \{v_{1,3},v_2,v_{1,3}^-\} $.
\end{proof}

\end{document}